\documentclass[runningheads]{llncs}

\usepackage[utf8]{inputenc}
\usepackage{amssymb,amsmath}
\usepackage{graphicx,subfigure,tabularx}
\usepackage{enumerate}
\usepackage{color}
\usepackage{times}
\usepackage{wrapfig}
\usepackage[activate=normal]{pdfcprot}
\usepackage{paralist}
\usepackage{enumitem}
\usepackage{hyperref}

\newcommand{\SP}[1]{\textcolor{red}{#1}}

\newcommand{\comm}[1]{}

\newcommand{\df}{\textbf}

\let\doendproof\endproof
\renewcommand\endproof{~\hfill\qed\doendproof}

\newcommand{\prob}[1]{\text{\textsc{#1}}}
\newcommand{\NP}[0]{\texttt{NP}}

\newcommand{\MLCM}{\textsc{MLCM}}
\newcommand{\PrMLCM}{\textsc{Proper}-\textsc{MLCM}}
\newcommand{\MLCMP}{\mbox{\textsc{MLCM-P}}}
\newcommand{\PrMLCMP}{\textsc{Proper}-\mbox{\textsc{MLCM-P}}}
\newcommand{\MLCMPA}{\mbox{\textsc{MLCM-PA}}}

\begin{document}
\date{}

\title{Metro-Line Crossing Minimization:\\ Hardness, Approximations, and
Tractable Cases}

\author{Martin Fink \inst1 \and Sergey Pupyrev \inst2}

\institute{Lehrstuhl f\"{u}r Informatik I, Universit\"{a}t
W\"{u}rzburg, Germany. \and Department of Computer Science, University of Arizona, USA.}

\maketitle
\begin{abstract}
  Crossing minimization is one of the central problems in graph
  drawing. Recently, there has been an increased interest in the
  problem of minimizing crossings between paths in drawings of graphs.
  This is the \emph{metro-line crossing minimization} problem
  (MLCM): Given an embedded graph and a set $L$ of simple paths, called
  \emph{lines}, order the lines on each edge so
  that the total number of crossings is minimized.
  So far, the complexity of \MLCM{} has been an open problem. In
  contrast, the problem variant in which line ends must be placed
  in outermost position on their edges (\MLCMP{}) is known to be NP-hard.

  Our main results answer two open questions: (i) We show that \MLCM{} is
  NP-hard. (ii) We give an $O(\sqrt{\log |L|})$-appro\-xi\-ma\-tion
  algorithm for \MLCMP{}.
\end{abstract}

\section{Introduction}
In metro maps and transportation networks, some edges, that is,
railway tracks or road segments, are used by several lines. Usually,
 lines that share an edge are drawn individually
along the edge in distinct colors; see Fig.~\ref{fig:example}.
Often, some lines must cross, and one normally wants to have
as few crossings of metro lines as possible. In the
\emph{metro-line crossing minimization} problem (\MLCM{}), the goal is
to order different metro-lines along each edge of the underlying
network, so that the total number of crossings is minimized.  Although
the problem has been studied~\cite{benkert07}, many questions remain open.

\begin{figure}[h]
  \centering
    \includegraphics[height=4cm]{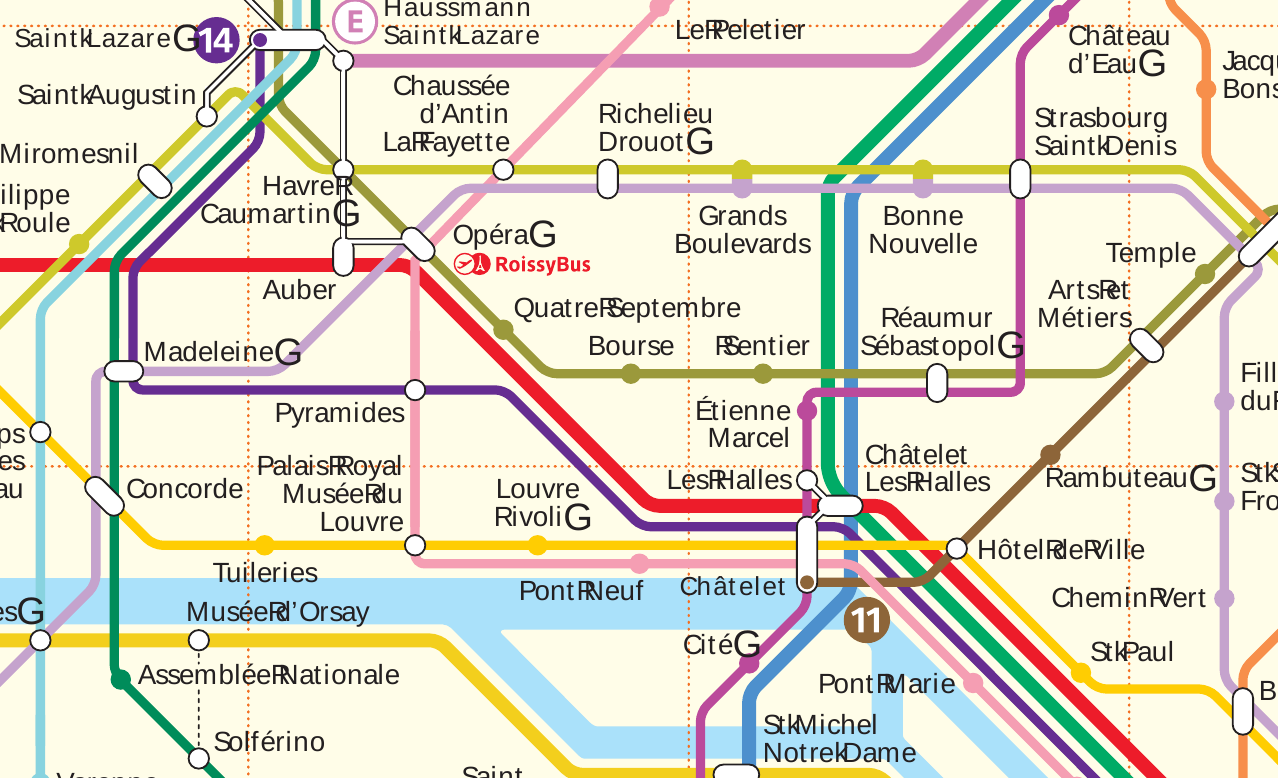}
  \caption{A part of the official metro map of Paris.}
    \label{fig:example}
\end{figure}

Apart from the visualization of metro maps,
the problem has various applications
including the visual representation of biochemical pathways.
In very-large-scale integration (VLSI) design, there is the closely
related problem of minimizing intersections between nets (physical
wires)~\cite{groeneveld89a,mareksadowska95}.
Net patterns with fewer crossings 
have better electrical
characteristics and require less area.  In graph drawing, the number of
edge crossings is one of the most important aesthetic criteria. In
\emph{edge bundling}, groups of edges are drawn close
together---like metro lines---emphasizing the structure of the
graph; minimizing crossings between parallel edges
arises as a subproblem~\cite{pupyrev11}.

\paragraph{Problem Definitions.}
The input is an embedded graph $G = (V, E)$ and
a set $L = \{l_1, \dots, l_{|L|}\}$ of simple paths in $G$.
We call $G$ the \df{underlying network}, the
vertices \df{stations}, and the paths \df{lines}.
The endpoints $v_0, v_k$ of a line $(v_0, \dots, v_k)\in~L$ are
\df{terminals}, and the vertices $v_1, \dots, v_{k-1}$ are
\df{intermediate stations}. For each edge $e = (u,v) \in E$, let $L_{e}$
be the set of lines passing through~$e$.

Following previous work~\cite{argyriou09,nollenburg09}, we use the
\emph{k-side} model;
each station $v$ is represented by a polygon with $k$ sides, where $k$ is
the degree of $v$ in $G$; see Fig.~\ref{fig:pc}. Each side of the polygon
is called a \df{port} of $v$ and corresponds to an incident edge $(v, u) \in E$.
A line $(v_0, \dots, v_k)$ is represented by a polyline starting at a port
of $v_0$ (on the boundary of the polygon), passing through two ports of $v_i$ for
$1 \le i < k$, and ending at a port of $v_k$. For each port of $u \in V$
corresponding to $(u, v) \in E$, we define the \df{line order} $\pi_{uv} = (l_1
\dots l_{|L_{uv}|})$ as an ordered sequence of the lines in $L_{uv}$, which
specifies the clockwise order at which the lines $L_{uv}$ are connected to the
port of $u$ with respect to the center of the polygon.
Note that there are two different line orders $\pi_{uv}$ and
$\pi_{vu}$ on any edge $(u,v)$ of the network. A \df{solution}, or a \df{line layout}, specifies line orders $\pi_{uv}$ and $\pi_{vu}$
for each edge $(u, v) \in E$.

\begin{figure}[t]
    \centering
		\hfill
    \includegraphics[page=3]{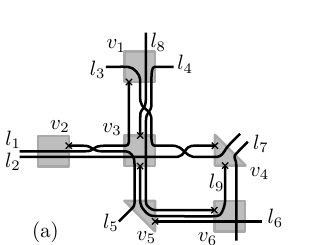}
		\hfill
    \includegraphics[page=4]{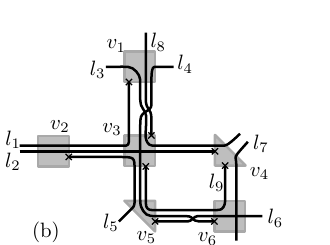}
		\hfill
    \caption{9 lines on an underlying network of 6 vertices and 9
      edges. (a) $\pi_{v_3v_4}=(l_3,l_2)$ and $\pi_{v_3v_1}=(l_1,
      l_8,l_4, l_3)$.  The lines $l_3$ and $l_4$ have an unavoidable
      edge crossing on $\left\{ v_1, v_3 \right\}$. In contrast, the
      crossing of $l_2$ and $l_3$ on $\{v_3,v_4\}$ is avoidable. In
      $v_3$ there is an unavoidable vertex crossing of the lines $l_2$
      and $l_8$.  As the vertex crossing of $l_2$ and $l_5$ in $v_3$
      is avoidable the solution is not feasible.  (b)~A feasible
      solution satisfying the periphery condition.}
    \label{fig:pc}
\end{figure}

A line crossing is a crossing between polylines corresponding to a pair of lines.
We distinguish two types of crossings; see Fig.~\ref{fig:pc}(a).
An \df{edge crossing} between lines $l_1$ and $l_2$ occurs whenever
$\pi_{uv} = (\dots l_1 \dots l_2 \dots)$ and $\pi_{vu} = (\dots l_1 \dots l_2 \dots)$
for some edge $(u, v) \in E$. We now
consider the concatenated cyclic sequence $\pi_u$ of the orders $\pi_{uv_1},
\dots, \pi_{uv_k}$, where $(u, v_1), \dots, (u, v_k)$ are the edges incident to
$u$ in clockwise order. A \df{vertex crossing} between $l_1$
and $l_2$ occurs in $u$ if $\pi_u = (\dots l_1 \dots l_2 \dots l_1 \dots l_2
\dots)$. 
Intuitively, the lines change their
relative order inside $u$. A crossing is called \df{unavoidable} if
the lines cross in any line layout;
otherwise it is \df{avoidable}.
A crossing is unavoidable if neither $l_1$ nor $l_2$ have a terminal on
their common subpath and the lines split on both ends of this subpath
in such a way that their relative order has to change; see Fig.~\ref{fig:pc}.
Following previous work, we insist that (i) \emph{avoidable
vertex crossings are not allowed} in a solution, that is, these crossings are not hidden
below a station symbol, and (ii) \emph{unavoidable vertex crossings
are not counted} since they occur in any solution.

A pair of lines may share several common
subpaths, and the lines may cross multiple times on the subpaths.
For simplicity of presentation, we assume that there is at most one common
subpath of two lines. Our results do, however, also hold for the general case
as every common subpath can be considered individually.

\paragraph{Problem variants.}
Several variants of the problem have been considered in the literature.
The original metro-line crossing minimization problem is formulated as follows.

\begin{problem}[\df{\MLCM{}}]
For a given instance $(G, L)$, find a line layout with the minimum
number of crossings.
\end{problem}

In practice, it is desirable to avoid gaps between
adjacent lines; to this end, every line is drawn so that
it starts and terminates at the topmost or bottommost end of a port; see
Fig.~\ref{fig:pc}(b).
In fact, many manually created maps follow this \df{periphery
condition} introduced by Bekos~et~al.~\cite{bekos08}.
Formally, we say that a line order $\pi_{uv}$ at the port of $u$
satisfies the periphery condition if $\pi_{uv} = (l_1 \dots l_p \dots l_q \dots l_{|L_{uv}|})$,
where $u$ is a terminal for the lines $l_1, \dots, l_p, l_q, \dots, l_{|L_{uv}|}$ and
$u$ is an intermediate station for the lines $l_{p+1}, \dots, l_{q-1}$.
The problem is known as \emph{\MLCM{} with periphery condition}.

\begin{problem}[\df{\MLCMP{}}]
For a given instance $(G, L)$, find a line layout, subject to the periphery
condition on any port, with the minimum number of crossings.
\end{problem}

In the special case of \MLCMP{} with \emph{side assignment}
(\df{\MLCMPA{}}), the input additionally
specifies for each line end on which side of its port it terminates;
N{\"o}llenburg~\cite{nollenburg09} showed that \MLCMPA{} is computationally equivalent to the version
of \MLCM{} in which all lines terminate at vertices of degree one.


As \MLCM{} and \MLCMP{} are \NP{}-hard even for very simple
networks, we introduce the additional
constraint that no line is a subpath of another line.
Indeed, this is often the case for bus and metro transportation
networks; if, however, there is a line that is a subpath of a longer
line then one can also visualize it as a part of the longer line.
We call the problems with this new restriction \PrMLCM{} and \PrMLCMP{}.

\paragraph{Previous Work.}
Metro-line crossing minimization was initiated by
Benkert~et~al.~\cite{benkert07}.
They described a quadratic-time algorithm
for \MLCM{} when the underlying network consists of a single edge with
attached leaves. The complexity status of \MLCM{} has been open.
As far as we are aware, this is the only known result on \MLCM{} so
far.

Bekos~et~al.~\cite{bekos08} studied \MLCMP{} and proved
that the variant is \NP{}-hard on paths. Motivated by the
hardness, they introduced the variant \MLCMPA{} and studied the problem on simple
networks.
Later, polynomial-time algorithms for \MLCMPA{} were found with gradually
improving running time by Asquith~et~al.~\cite{asquith08},
Argyriou~et~al.~\cite{argyriou09}, and N{\"o}llenburg~\cite{nollenburg09}, until
Pupyrev~et~al.~\cite{pupyrev11} presented  a linear-time algorithm.
Asquith et~al.~\cite{asquith08} formulated \MLCMP{} as an integer linear program that
finds an optimal solution for the problem on general graphs. Note that
in the worst case this approach requires exponential time.
Recently, Fink and Pupyrev studied a variant of \MLCM{} in which whole
blocks of lines may cross~\cite{fink+pupyrev13}.

In the circuit
design community (VLSI), Groeneveld~\cite{groeneveld89a} considered the problem of
adjusting the routing so as to minimize crossings between the pairs of nets,
which is equivalent to \MLCMPA{}, and suggested an algorithm for general graphs.
Another method for graphs of maximum degree four was given in~\cite{chen98}.
Marek-Sadowska~et~al.~\cite{mareksadowska95} considered a related problem
of distributing the line crossings among edges of the underlying graph in order
to simplify the net routing.

\begin{table}[t]
    \centering
    \begin{tabular}{|l|l|l|l|}
        \hline
        problem &         graph class &   result   &  reference \\
        \hline
        \hline
				\MLCM{} & caterpillar & \NP{}-hard & Thm.~\ref{thm:mlcm-h-hardness}\\
        \MLCM{} & single edge & $O(|L|^2)$-time algorithm & \cite{benkert07} \\
				\MLCM{} & general graph & crossing-free test & Thm.~\ref{thm:mlcm-planar}\\
        \hline
        \MLCMP{} & path & \NP{}-hard & \cite{argyriou09}\\
        \MLCMP{} & general graph & ILP & \cite{asquith08}\\
        \MLCMP{} & general graph & $O(\sqrt{\log |L|})$-approximation &
				Thm.~\ref{thm:approx}\\
        \MLCMP{} & general graph & crossing-free test & Thm.~\ref{thm:2sat}\\
        \hline
		\PrMLCMP{} & general graph with consistent lines & $O(|L|^3)$-time
        algorithm&
				Thm.~\ref{thm:consist}\\
        \hline
        \MLCMPA{} & general graph & $O(|V| + |E| + |V||L|)$-time & \cite{pupyrev11}\\
        \MLCMPA{} & general graph & crossing-free test & \cite{nollenburg09} \\
        \hline
    \end{tabular}
    \smallskip
    \caption{Overview of results for the metro-line crossing
    minimization problem.}
    \label{table:res}
\end{table}

\paragraph{Our Results.}
\label{sec:contributions}
Table~\ref{table:res} summarizes our contributions and previous results.
We first prove that the unconstrained variant \MLCM{}
is \NP{}-hard even on caterpillars (paths with attached leaves), thus, answering
an open question of Benkert et~al.~\cite{benkert07} and N{\"o}llen\-burg~\cite{nollenburgThesis}.
As crossing minimization is hard, it is natural to ask whether there exists a
\df{crossing-free} solution. We show that there is a crossing-free solution if and
only if there is no pair of lines forming an unavoidable crossing.

We then study \MLCMP{}. Argyriou~et~al.~\cite{argyriou09} and N{\"o}llenburg~\cite{nollenburgThesis}
asked for an approximation algorithm. To this end, we develop a 2SAT model for
the problem. Using the model we get an $O(\sqrt{\log |L|})$-approximation
algorithm for \MLCMP{}.
This is the first approximation algorithm in the context
of metro-line crossing minimization.
We also show how to find a crossing-free solution (if it exists) in polynomial time.
Moreover, we prove that \MLCMP{} is fixed-parameter tractable with respect to
the maximum number $k$ of allowed crossings by using the
fixed-parameter tractability of 2SAT.

We then study the new variant \PrMLCMP{} and show how to solve it on caterpillars,
\emph{left-to-right trees} (considered in~\cite{bekos08,argyriou09}),
and other instances described in Section~\ref{sec:p-mlcm-p}.
An optimal solution can be found by applying a maximum flow algorithm
on a certain graph.
This is the first polynomial-time exact
algorithm for the variant in which avoidable crossings may be presented in an
optimal solution.

Finally, we consider practical aspects of the proposed algorithms. We show how
to find a (not necessarily optimal) line layout in the cases where some of the
required constraints are not fulfilled.


\section{The \MLCM{} Problem}
We begin with the most flexible problem variant \MLCM{}, and show that
it is hard to decide whether there is a solution
with at most $k>0$ crossings, even if the underlying network is a
caterpillar.
In contrast, we give a polynomial-time algorithm for deciding whether there exists
a crossing-free solution.

\subsection{NP-Hardness}
\label{sec:mlcm-hardness}

\begin{theorem}
\label{thm:mlcm-h-hardness}
  \MLCM{} is \NP{}-hard on caterpillars.
\end{theorem}
\begin{proof}
  We prove hardness by reduction from \MLCMP{} which is known to be
  \NP{}-hard on paths~\cite{argyriou09}. Suppose we have an instance
  of \MLCMP{} consisting of a path $G = (V,E)$ and lines $L$ on the path. We want
  to decide whether it is possible to order the lines with periphery
  condition and at most $k$ crossings.

	We create a new underlying network $G' = (V',E')$ by adding some
	vertices and edges to $G$.
	We assume that $G$ is embedded along a horizontal line and specify new positions
	relative to this line.
  For each edge $e = (u,v) \in E$, we add vertices $u_1, u_2,
  v_1$, and $v_2$ and edges $(u,u_1)$, $(u,u_2)$, $(v,v_1)$, and
  $(v,v_2)$ such that $v_1$ and $u_1$ are above the path and
  $v_2$ and $u_2$ are below the path. Next, we add $\ell = |L|^2$ lines from
  $u_1$ to $v_2$, and $\ell$ lines from $u_2$ to $v_1$ to $L' \supseteq
	L$; see Fig.~\ref{fig:red_cross_insertion}. We call the added structure the
	\df{red cross} of $e$, the added lines \df{red lines}, and the lines of $L$ old
	lines. We claim that there is a
  number $K$ such that there is a solution of \MLCMP{} on $(G,L)$ with
  at most $k$ crossings if and only if there is solution of
  \MLCM{} on $(G',L')$ with at most $k+K$ crossings.

  \begin{figure}[t]
    \centering
    \includegraphics[page=1]{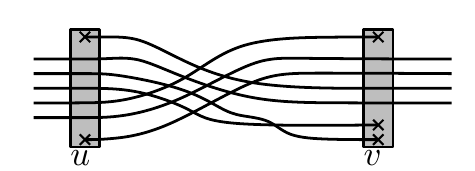}
    \hfill
    \includegraphics[page=2]{pics/red_cross}
    \caption{(a) \MLCMP{}-solution on edge
    $(u,v)$. (b) Insertion of a red cross into the solution with minimum
    number of additional crossings.}
    \label{fig:red_cross_insertion}
  \end{figure}

  Let $e = (u,v) \in E$ be an edge of the path, and let $l\in L_e$ be a line on
  $e$. If $l$ has its terminals on $u$ and $v$, that is,
  completely lies on $e$, it never has to cross in $G$ or $G'$; hence,
  we assume such lines do not exist. Assume $l$ has none of
  its terminals on $u$ or $v$. It is easy to see that it has to
  cross all $2\ell$ lines of the red cross of $e$. Finally, suppose $l$ has just one terminal
  at a  vertex of $e$, say on $u$. If the terminal is above the
  edge $(u, u_1)$ then it has to cross all red lines from $u_2$ to
  $v_1$ but can avoid the red lines from $u_1$ to $v_2$, that is,
  $\ell$ crossings with red lines are necessary.
  Symmetrically, if the terminal is below $(u,u_2)$ then
  only the $\ell$ crossings with the red lines from $u_1$ to
  $v_{2}$ are necessary. If the terminal is between the edges
  $(u,u_1)$ and $(u,u_2)$ then all $2\ell$ red edges must be crossed.
  There are, of course, always $\ell^2$ unavoidable internal crossings of the
	red cross of $e$.

  Let $\ell_e = \ell_e^t + \ell_e^m$ be the number of lines on
  $e$, where $\ell_e^t$ and $\ell_e^m$ are the numbers of lines on
  $e$ that do or do not have a terminal at $u$ or $v$, respectively.
  In any solution there are at least $\ell_e^t
  \cdot \ell + 2 \cdot \ell_e^m \cdot \ell + \ell^2$ crossings on
  $e$ in which at least one red line is involved. It is easy to see
  that placing a terminal between red lines leaving towards a leaf
  never brings an advantage. On the other hand, if just a single line
  has an avoidable crossing with a block of red lines, the number of
  crossings increases by $\ell = |L|^2$, which is more than the
  number of crossings in any solution for $(G,L)$ without double crossings. Hence, any
  optimal solution of the lines in $G'$ has no avoidable
  crossings with red blocks and, therefore, satisfies the periphery
  condition; thus, after deleting the added edges and red lines, we
  get a feasible solution for \MLCMP{} on $G$.

  Let $K:= |E|\cdot \ell^2 + \sum_{e \in E}\left( \ell_e^t
  + 2 \ell_e^m\right) \cdot \ell $ be the minimum number of crossings
  with red lines involved on $G'$. Suppose we have an \MLCM{}-solution on
  $G'$ with at most $K+k$ crossings. Then, after deleting the red
  lines, we get a feasible solution for \MLCMP{} on $G$ with at most
  $k$ crossings. On the other hand, if we have an \MLCMP{}-solution on
  $G$ with $k$ crossings, then we can insert the red lines with just
  $K$ new crossings: Suppose we want to insert the block of red lines
  from $u_1$ to $v_2$ on an edge $e = (u,v) \in E$. We start by
  putting them immediately below the lines with a terminal on the top
  of $u$. Then we cross all lines below until we see the first line
  that ends on the bottom of $v$ and, hence, must not be crossed by
  this red block. We go to the right and just keep always directly
  above the block of lines that end at the bottom side of $v$; see
  Fig.~\ref{fig:red_cross_insertion}.
  Finally, we reach $v$ and have not created any avoidable crossing.
  Once we have inserted all blocks of red lines, we get a solution for
  the lines on $G$ with exactly $K+k$ crossings.
\end{proof}

\subsection{Recognition of Crossing-Free Instances}
\label{sec:mlcm-planar}
Given an instance of \MLCM{}, we want to check whether there exists a
solution without any crossings. If there exists such a crossing-free
solution then there cannot be a pair of lines with an unavoidable
crossing. We show that this condition is already sufficient.

Consider a pair of lines $l_1,l_2$ with a common subpath $P=(v, v_1,
\ldots,  u_1, u)$; see
Fig.~\ref{fig:unavoidable}. Suppose the lines \df{split} at $v$, that is, neither $l_1$ nor $l_2$ terminates at $v$.
Since vertex crossings are not allowed in our model, there is a unique order
between $l_1$ and $l_2$ at the port $vv_1$ in any solution of \MLCM{}.
Furthermore, in any crossing-free solution, the relative order of $l_1$ and
$l_2$ is the same on all ports.

We arbitrarily fix a direction for each edge of the underlying
network. For an edge $e=(u, v)\in E$ directed from $u$ to $v$ and for a pair of lines $l_1, l_2 \in L_{uv}$,
we say that $l_1$ is \df{above} $l_2$ if $\pi_{uv} = (\dots l_1 \dots l_2 \dots)$ in any crossing-free solution.
Otherwise, if $\pi_{uv} = (\dots l_2 \dots l_1 \dots)$ in any crossing-free solution, we say that $l_1$ is
\df{below} $l_2$. Note that on some other edge $e'$, $l_1$ may be below $l_2$,
depending on the direction of $e'$.  We say that a line $l$ \textbf{lies} between
$l_1$ and $l_2$ if $l_1$ is above $l$ and $l$ is above $l_2$ on $e$.
First, a useful observation.

\textbf{Observation.} The lines $l_1, l_2$ have an unavoidable crossing if and only if
they split in such a way that, on some edge $e$, $l_1$ has to be above $l_2$
and at the same time $l_2$ has to be above $l_1$.

\begin{figure}[h]
      \centering
      \includegraphics[page=1]{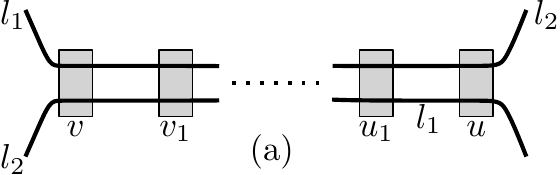}
			\hfill
      \includegraphics[page=2]{pics/ex3}
  \caption{(a)~$l_1$ is above $l_2$ at the port $vv_1$ but below
    $l_2$ at the port $uu_1$; the crossing is unavoidable.
  (b)~A terminal on a common subpath at $u$; the crossing is avoidable.}
  \label{fig:unavoidable}
\end{figure}

We assume that no line is a subpath of another line as
a subpath can be reinserted parallel to the longer line in a crossing-free solution.
Consider a pair of lines $l_1, l_2$ whose common subpath $P$ starts
in $u$ and ends in $v$. If $u$ (similarly, $v$)
is a terminal neither for $l_1$ nor $l_2$ then there is a unique relative order
of the lines along $P$ in any crossing-free solution;
see Fig.~\ref{fig:unavoidable}.
Hence, we assume $u$ is a terminal for $l_1$, $v$ is a terminal for $l_2$, and
we call such a pair \df{overlapping}.
Suppose there is a \df{separator} for $l_1$ and $l_2$, that is,
a line $l$ on the subpath of $l_1$ and $l_2$ that has to be below
$l_1$ and above $l_2$ (or the other way round) as shown in
Fig.~\ref{fig:separating-line}. Then, $l_1$ has to
be above $l_2$ in a crossing-free solution. The only remaining case is a pair of lines without a separator.

\begin{figure}[t]
	\begin{minipage}[t]{0.39\textwidth}
		\centering
		\includegraphics{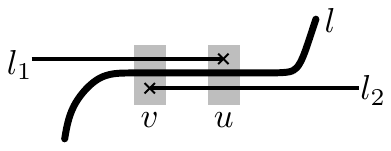}
		\caption{A separator $l$ of lines $l_1$ and $l_2$.}
		\label{fig:separating-line}
	\end{minipage}
	\hfill
	\begin{minipage}[t]{0.59\textwidth}
		\centering
		\includegraphics{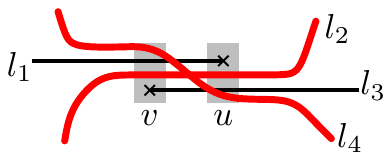}
		\caption{Unavoidable crossing of 2 separators of $l_1$ and
		$l_3$.
		}
		\label{fig:mlcm-no-cycle-length-4}
	\end{minipage}
\end{figure}

\newcounter{mergelemmacounter}
\setcounter{mergelemmacounter}{\value{lemma}}
\newcommand{\mergelemma}{
  Let $l_1, l_2$ be a pair of overlapping lines without a separator for
  which the number of edges of the common subpath is minimum.
  If there exists a crossing-free solution then 
  there is also a crossing-free solution in which $l_1$ and $l_2$ are immediate neighbors
  in the orders on their common subpath.
}

\begin{lemma}
  \mergelemma
  \label{lemma:merge-no-separator}
\end{lemma}
\begin{proof}
  First, let us show that no line has its terminal in an intermediate station
	of the common subpath of $l_1$ and $l_2$. Suppose there is such a line
	$l$. Then $l$ forms an overlapping pair with either $l_1$
	or $l_2$---say $l_1$ without loss of generality---, whose common subpath is
	shorter than the one of $l_1$ and $l_2$.  Hence, there is a
  separator $l'$ of $l_1$ and $l_1$, which also separates $l_1$ and $l_2$ in contradiction to
  the choice of $l_1$ and $l_2$.

  Now, suppose there is no crossing-free solution in which $l_1$ and $l_2$ are
	immediate neighbors on their complete subpath. We fix a crossing-free solution in
  which the number of lines lying between $l_1$ and $l_2$ is minimal and
  suppose $l_1$ is above $l_2$.

  Let $l$ be the topmost line that lies between $l_1$ and $l_2$ in the
  solution and overlaps with
  $l_1$ (or symmetrically, the bottommost that overlaps with
  $l_2$). By modifying the ordering, considering
  lines above $l$ and below $l_1$, and some lines
  not overlapping with $l_1$, it is possible to reroute $l$ so that it does not lie between
  $l_1$ and $l_2$, in contradiction to the choice of the solution. To this end, let $S \supseteq
  \left\{ l_1,l \right\}$ be the smallest superset of $l$ and $l_1$ such that for
	any pair of lines $l',l'' \in S$ any line that lies between $l'$ and $l''$ in the
	solution is also contained in $S$; see
  Fig.~\ref{fig:separating-line-lemma-1-appendix}. Note that no pair of lines
  in $S$ has a separator as this would also be a separator
  for $l_1$ and $l_2$.

  \begin{figure}[h]
    \subfigure[Lines $S$ between $l_1$ and $l$.]{
      \centering
      \includegraphics[page=1]{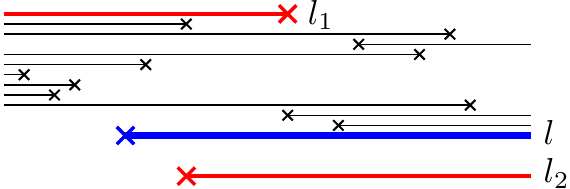}
    \label{fig:separating-line-lemma-1-appendix}
  }
    \hfill
    \subfigure[Lines between $l_1$ and $l$ reordered.]{
      \includegraphics[page=2]{pics/separating-line-lemma}
    \label{fig:separating-line-lemma-2-appendix}
  }
  \caption{Rerouting lines between $l_1$ and $l_2$. All shown lines
    share a subpath, which is shown in the drawings. At any position
  with a terminal, there is a node (not drawn).}
  \label{fig:separating-line-lemma-appendix}
  \end{figure}

  If $S=\left\{ l_1,l \right\}$ we can reroute $l$ to be above
  $l_1$. Otherwise, we apply the following procedure. For any pair of
  overlapping lines in this set that are immediate neighbors, that is,
  there is no line lying in between, we
  reroute the right one to be immediately above the left one, which is
  possible as there is no separating line. Eventually, $l$ will be
  above $l_1$, otherwise there would still be steps to be performed.
  Hence, we can create a solution in which there is at least one line
  less between $l_1$ and $l_2$, a contradiction.
\end{proof}

\newcounter{mlcmcfcounter}
\setcounter{mlcmcfcounter}{\value{theorem}}
\newcommand{\mlcmcf}{
  Any instance of \MLCM{} without unavoidable crossings has a
  crossing-free solution.
}

\begin{theorem}
  \mlcmcf
\label{thm:mlcm-planar}
\end{theorem}
\begin{proof}
Consider an instance of \MLCM{} without unavoidable crossings.
Using Lemma~\ref{lemma:merge-no-separator}, we can merge a pair
of overlapping lines  without a separator into a new line. The merging cannot
introduce an unavoidable crossing: Suppose there would be a line $l$ forming an
unavoidable crossing with the merged line $l'$ of $l_1$ and $l_2$. Indeed,
$l$ and $l'$ have to split on both ends with with different side
constraints. The splits have to be on different sides of the common subpath of
$l_1$ and $l_2$, otherwise there already was an unavoidable crossing of
$l$ with $l_1$ or $l_2$. From the splits we get relative orders of $l$ with
$l_1$ and $l_2$ such that either $l_1$ is above $l$ and $l$ is above $l_2$, or
$l_2$ is above $l$ and $l$ is above $l_1$. In both cases $l$ already was a
separator for $l_1$ and $l_2$ and we would not have merged them.

We iteratively perform merging steps until any overlapping pair has a
separator. Note that there might be multiple separators for a pair, but
all of them separate the pair in the same relative order;
otherwise, we would have a pair of separators with an unavoidable
crossing; see Fig.~\ref{fig:mlcm-no-cycle-length-4}.

After the merging steps, for any pair of lines sharing an edge, we either get a unique relative
order for crossing-free solutions, or the pair has a separator.

We now create a directed \textbf{relation graph} $G_{e}$ for any edge $e \in E$.
Vertices of the graph are the lines $L_{e}$ passing through $e$. Edges of
$G_e$ model the relative order of the lines in a crossing-free solution; we have
an edge $(l_1, l_2)$ (similarly, $(l_2, l_1)$) in $G_{e}$ if $l_1$ and
$l_2$ split in such a way that $l_1$ is above (below) $l_2$ in any crossing-free solution.

Let us prove that all relation graphs are acyclic.
Suppose there is a cycle in a relation graph $G_e$.
We choose the shortest cycle $C$. A cycle of
length~2 is equivalent to a pair of lines with an unavoidable crossing;
hence, such a cycle cannot exist.

  Now, suppose there is a cycle $C=(l_1,l_2,l_3)$ of length $3$. Lines
  $l_1$ and $l_2$ share a common subpath and split on one side in the
  order $(l_1,l_2)$. The splitting for realizing the edge $(l_2,l_3)$
  can not be realized on this subpath, otherwise we would also get the
  edge $(l_1,l_3)$. Similarly, the splitting for $(l_3,l_1)$ also can
  not be realized on the subpath. Hence, we have to distribute the two
  splittings to the two sides of the subpath, which is not possible
  without introducing an unavoidable crossing with $l_1$ or
  $l_2$; see Fig.~\ref{fig:mlcm-no-cycle-length-3-appendix}.

  Finally, if the cycle $C$ is longer then there exists a path
  $(l_1,l_2,l_3,l_4)$ of length four without chords.
  As there is no edge between $l_1$ and $l_3$, they have to be an overlapping pair
  and $l_2$ is a separator for them. On the other hand,
  $l_4$ is also a separator for $l_1$ and $l_3$, but separates them in
  another relative order. It is easy to see that there is an
  unavoidable crossing of $l_2$ and $l_4$, a contradiction; see
  Fig.~\ref{fig:mlcm-no-cycle-length-4}. Hence, the relation graphs are acyclic.

  \begin{figure}[h]
    \subfigure[]{
      \includegraphics[page=2]{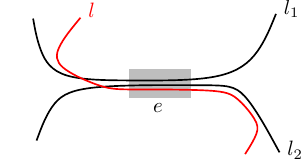}
    }
    \hfill
    \subfigure[]{
      \includegraphics[page=4]{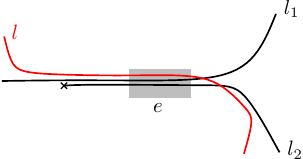}
    }
    \caption{There is no cycle of 3 lines, if $l_1$ and $l_2$ split
    (a) on both sides or (b) only on one side, without an unavoidable
  crossing.}
    \label{fig:mlcm-no-cycle-length-3-appendix}
  \end{figure}

Now, in a relation graph $G_e$ for any pair of lines $l_1,l_2 \in L_e$,
there is either a directed edge connecting them in $G_e$, or the lines are overlapping.
Since $G_e$ is acyclic, there exists a unique topological ordering of the lines $L_e$.
We get a crossing-free solution by using the ordering for every edge. As the
relative order of any pair of lines is the same for all edges, there cannot be a crossing.
\end{proof}

The proof yields an algorithm for finding a crossing-free solution. It takes
$O(|L|^2 |E|)$ time for deleting subpaths as well as iteratively merging the shortest
unseparated overlapping pair. Finally, we can get the relative order of each pair
of lines on all edges in $O(|L|^2|E|)$ time and order the lines on all edges.
Hence, after reinserting the deleted or merged lines, we get a
crossing-free solution in $O(|L|^2|E|)$ time.

%


\section{The \MLCMP{} Problem}
\label{sec:mlcm-p}

\label{sec:mlcm-p-2sat}
Let $\left( G = (V,E), L \right)$ be an instance of \MLCMP{}. Our goal
is to decide for each line end on which side of its
terminal port it should lie. For convenience, we arbitrarily choose
one side of each port and call it ``top'', the opposite side is
called ``bottom''.
For each line $l$ starting at vertex $u$ and ending at
vertex $v$, we create binary variables $l_{u}$ and
$l_{v}$, which are true if and only if $l$ terminates at
the top side of the respective port.
We formulate the problem of finding a truth assignment that minimizes
the number of crossings as a 2SAT instance for the given instance of
\MLCMP{}. Note that Asquith~et~al.~\cite{asquith08} already used 2SAT
clauses as a tool for developing their ILP for \MLCM{}, where the
variables represent above/below relations between line ends.
In contrast, in our model a variable directly represents the
position of a line on the top or bottom side of a port. We first prove
a simple property of lines.

\begin{lemma}
	Let $l,l'$ be a pair of lines sharing a terminal. We can transform any solution in
		which $l$ and $l'$ cross to a solution with fewer crossings in which the
		lines do not cross.
	\label{lemma:2-terminals-same-vertex}
\end{lemma}

\begin{proof}
Assume $l$ and $l'$ cross in a solution.
We switch the positions of line ends at the common terminal $v$ between
$l$ and $l'$ and reroute the two lines between the crossing's
position and $v$. By reusing the route of $l$ for $l'$ and vice versa,
the number of crossings does not increase. On the other
hand, the crossing between $l$ and $l'$ is eliminated.
\end{proof}

Let $l, l'$ be two lines whose common
subpath $P$ starts at vertex $u$ and ends at $v$. Observe that
terminals of $l$ and $l'$ that lie on $P$ can only be at $u$ or
$v$. If neither $l$ nor $l'$ has a terminal on $P$ then a crossing
of the lines does not depend on the positions of the terminals;
hence, we assume that there is at least one terminal at $u$ or $v$.
A possible crossing between $l$ and $l'$ is modeled by a 2SAT formula,
the \df{crossing formula}, consisting of at most two clauses. The
crossing formula evaluates to
true if and only if $l$ and $l'$ do not cross. For simplicity, we
assume that the top sides of the terminal ports of $u$ and $v$ are
located on the same side of $P$; see Fig.~\ref{fig:se}. If
it is not the case, a variable $l_u$ should be substituted with its inverse $\neg l_u$
in the formula. We consider several cases; see also
Fig.~\ref{fig:crossing-formulas}.

\begin{enumerate}[label=(f\arabic*)]
\item Suppose $u$ and $v$ are terminals for $l$ and intermediate stations for $l'$,
that is, $l$ is a subpath of $l'$. Then, $l$ does not
cross $l'$ if and only if both terminals of $l$ lie on the same side of $P$.
This is expressed by the crossing formula $(l_u \wedge l_v) \vee
(\neg l_u \wedge \neg l_v) \equiv (\neg l_u \vee l_v) \wedge
(l_u \vee \neg l_v)$, which may occur multiple times, caused by
a different $l'$.
\label{formula-subpath}

\item Suppose $u$ is a terminal for $l$ and intermediate for $l'$,
and $v$ is a terminal for $l'$ and intermediate for $l$. Then there is no
crossing if and only if both terminals lie on opposite sides of $P$.
This is described by the formula $(l_u \wedge \neg l_v') \vee
(\neg l_u \wedge l_v') \equiv (l_u \vee l_v') \wedge (\neg l_u \vee
\neg l_v')$.
\label{formula-overlap}

\item Suppose both $l$ and $l'$ terminate at the same vertex $u$ or $v$.
By Lemma~\ref{lemma:2-terminals-same-vertex}, a solution
of \MLCMP{} with a crossing of $l$ and $l'$ can be transformed into a
solution in which $l$ and $l'$ do not cross.
Hence, we do not introduce formulas in this case.
\label{formula-same end}

\item In the remaining case, there is only one terminal of $l$ and $l'$ on $P$.
Without loss of generality, let $l$ terminate at $u$. A crossing
is triggered by a single variable. Depending on the
fixed terminals or leaving edges at $v$ and $u$, we get the single clause
$(l_{u})$ or $(\neg l_{u})$. Note that the same clause can occur
multiple times, caused by different lines $l'$.
\label{formula-just-one}
\end{enumerate}

\begin{figure}[t]
	\centering
		\includegraphics[page=1]{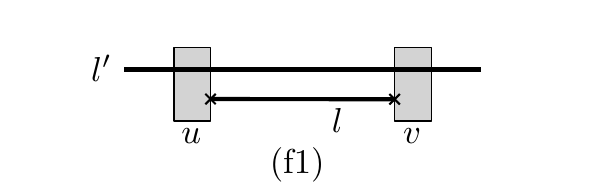}
		\includegraphics[page=2]{pics/crossing-formulas}
		\includegraphics[page=3]{pics/crossing-formulas}
		\includegraphics[page=4]{pics/crossing-formulas}
	\caption{Four cases for crossing formulas:
	(f1) $(l_u \wedge l_v) \vee (\neg l_u \wedge \neg l_v) \equiv (\neg l_u \vee
	l_v) \wedge (l_u \vee \neg l_v)$;
	(f2) $(l_u \wedge \neg l_v') \vee (\neg l_u \wedge l_v') \equiv (l_u \vee
	l_v') \wedge (\neg l_u \vee \neg l_v')$;
	(f3) crossing can always be removed;
	(f4) $(l_{u})$.
	}
	\label{fig:crossing-formulas}
\end{figure}

\begin{figure}[t]
  \centering
  \subfigure[An instance ($G,L$) of \PrMLCMP{}.]{
      \includegraphics[]{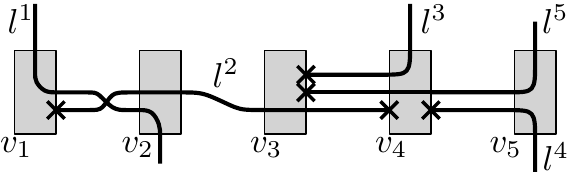}
    \label{fig:sea}}
    \hfill
    \subfigure[Graph $G_{ab}$ for the instance $(G,L)$.]{
      \includegraphics[width=0.4\textwidth]{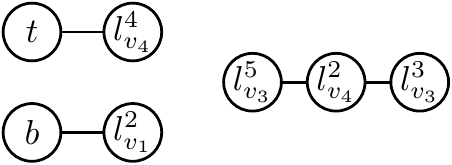}
    \label{fig:seb}}
    \caption{
      A small instance of \MLCMP. The generated 2SAT formulas are:
			$(l^{2}_{v_{1}})$ for the crossing of $l^1$ and $l^2$; $(\neg l_{v_4}^4)$
			for the crossing of $l^5$ and $l^4$;
    $(l_{v_4}^2 \vee l_{v_3}^3) \wedge (\neg l_{v_4}^2 \vee \neg
    l_{v_3}^3)$ for the crossing of $l^2$ and $l^3$;
    $(l_{v_4}^2 \vee l_{v_3}^5) \wedge (\neg l_{v_4}^2 \vee \neg
    l_{v_3}^5)$ for the crossing of $l^2$ and $l^5$.
    }
  \label{fig:se}
\end{figure}

\paragraph{Crossing-free solutions.}
Note that the 2SAT formulation of the problem yields an algorithm for deciding
whether there exists a crossing-free solution of an \MLCMP{}
instance. First, we check for unavoidable crossings by analyzing
every crossing formula individually. Second, the 2SAT model is
satisfiable if and only if there is a solution of the \MLCMP{} instance
without avoidable crossing. Since 2SAT can be solved in linear
time and there are at most $|L|^2$ crossing formulas, we conclude as
follows.

\begin{theorem}
\label{thm:2sat}
Deciding whether there exists a crossing-free solution for
\MLCMP{} can be accomplished in $O(|E||L|^2)$ time.
\end{theorem}

For \MLCM{} the existence of a crossing-free solution is equivalent to
the absence of unavoidable crossings. In contrast, there is no such
simple criterion for \MLCMP{}. Moreover, for any $k$, there is an
instance with $k$ lines such that any subset of $k-1$ lines
admits a crossing-free solution, while $k$ lines require at least one crossing;
see Fig.~\ref{fig:no-crossing-free-criterion}.

\begin{figure}[ht]
	\centering
    \includegraphics[height=3cm]{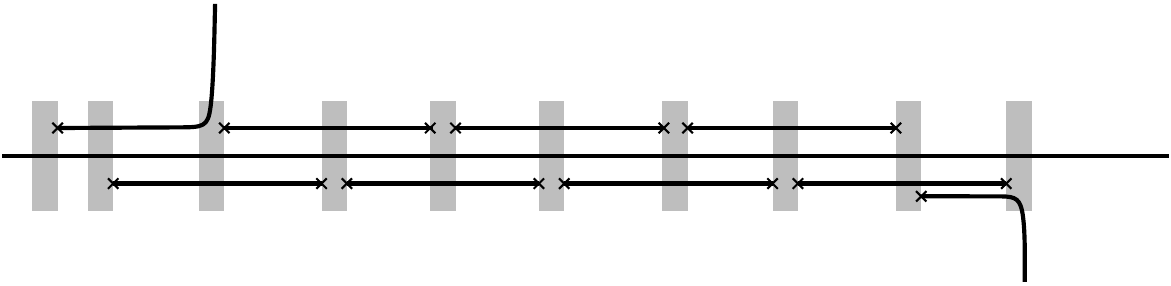}
  \caption{Example of an instance of \MLCMP{} that does not have a
  crossing-free solution. There is, however, no small substructure of
  lines that gives a contradiction to the assumed existence of a
crossing-free solution. Any proper subset of the lines allows a
crossing-free solution. Note that the example can easily be extended
to an arbitrary number of lines.}
  \label{fig:no-crossing-free-criterion}
\end{figure}

\paragraph{Fixed-parameter tractability.}
We can use the 2SAT model for obtaining a fixed-parameter tractable algorithm on
the number $k$ of allowed crossings. We must show that we can check in
$f(k) \cdot \mathrm{poly}(l)$ time whether there is a solution with at most $k$
avoidable crossings, where $f$ must be a computable function and
$l$ is the input size.

First, note that minimizing the number of crossings is
the same as maximizing the number of satisfied clauses in the
corresponding 2SAT instance. Maximizing the number of satisfied
clauses, or solving the \prob{MAX-2SAT} problem, is \NP{}-hard.

However, the problem of deciding whether it is possible to remove
a given number $k$ of $m$ 2SAT clauses so that the formula becomes
satisfiable is fixed-parameter tractable with respect to the
parameter~$k$ \cite{Razgon2009}. This yields the following
observation.

\begin{theorem}
  \MLCMP{} is fixed-parameter tractable with respect to the maximum
  allowed number of avoidable crossings.
  \label{theorem:mlcp-p-fpt-crossing-number}
\end{theorem}
\begin{proof}
  We show that the SAT formula can be made satisfiable by removing at
  most $k$ clauses if and only if there is a solution with at
  most $k$ crossings.

  First, suppose it is possible to remove at most $k$ clauses from the
  2SAT model so that there is a truth assignment satisfying all
  remaining clauses. Fix such a truth assignment, and consider the
  corresponding assignment of sides to the terminals. Any
  crossing leads to an unsatisfied clause in the SAT formula,
  and no two crossings share an unsatisfied clause. Hence, we have
  a side assignment that causes at most $k$ crossings.

  Now, we assume that there is an assignment of sides for all
  terminals that causes at most $k$ crossings.
  We know that in the corresponding truth assignment for all pairs of clauses of
  the SAT model at most one is unsatisfied. Hence, there are at most
  $k$ unsatisfied clauses since any crossing just leads to a single
  unsatisfied clause. The removal of these clauses creates a new,
  satisfiable formula.
\end{proof}
Using the $O(15^k k m^3)$-time algorithm for 2SAT~\cite{Razgon2009}
our algorithm has a running time of $O(15^k \cdot k \cdot |L|^6 +
|L|^2|E|)$.

\paragraph{Approximating \MLCMP{}.}
The proof of Theorem~\ref{theorem:mlcp-p-fpt-crossing-number} yields
that the number of crossings in a
crossing-minimal solution of \MLCMP{} equals the minimum number of
clauses that we need to remove from the 2SAT formula in order to
make it satisfiable. Furthermore, a set of
$k$ clauses, whose removal makes the 2SAT formula satisfiable,
corresponds to an \MLCMP{} solution with at most $k$ crossings. Hence,
an approximation algorithm for the problem of making a 2SAT formula
satisfiable by removing the minimum number of clauses (also called
\prob{Min 2CNF Deletion}) yields an approximation for \MLCMP{} of the
same quality. As there is an $O(\sqrt{\log m})$-approximation
algorithm for \prob{Min 2CNF deletion} \cite{Agarwal05}, we have the
following result.
\begin{theorem}
\label{thm:approx}
There is an $O(\sqrt{\log |L|})$-approximation algorithm for \MLCMP{}.
\end{theorem}

\section{The \PrMLCMP{} Problem}
\label{sec:p-mlcm-p}

In this section we consider the \PrMLCMP{} problem, where no line in $L$
is a subpath of another line. First we focus on graphs whose
underlying network is a caterpillar. There, the top and bottom sides
of ports are given naturally; see Fig.~\ref{fig:se}.

Based on the 2SAT model described in the previous section, we construct a
graph $G_{ab}$, which has a vertex $l_u$ for each variable of the model and two additional
vertices $b$ and~$t$. Since no line is a subpath of another line,
our 2SAT model has only the two types of crossing formulas
\ref{formula-overlap} and \ref{formula-just-one}; compare
Section~\ref{sec:mlcm-p-2sat}. For case~\ref{formula-overlap}, we create an edge $(l_u, l'_v)$.
The edge models a possible crossing between lines $l$ and $l'$; that is, the
lines cross if and only if $l$ terminates on top (bottom) of $u$ and $l'$
terminates on top (bottom) of $v$.  For a crossing formula of type
$(l_u)$ (case~\ref{formula-just-one}), we add an edge $(b, l_u)$ to $G_{ab}$; similarly, we add an edge $(t,
l_u)$ for a formula $(\neg l_u)$; see Fig.~\ref{fig:seb} for an example.

Any truth assignment to the variables
is equivalent to a $b$-$t$ cut in $G_{ab}$, that is, a cut separating $b$ and $t$.
Indeed, any edge in the graph models the fact that two lines should not be
assigned to the same side as they would cause a crossing otherwise.
Hence, any line crossing corresponds to an \emph{uncut} edge.
Therefore, to find a line layout with the minimum number of crossings,
we need to solve the known \prob{MIN-UNCUT} problem:

\begin{problem}[MIN-UNCUT]
Given a graph, divide its vertices into two partitions $S_t, S_b$ so that
the number of uncut edges ($(v, u)$ with $v,u\in S_t$ or $v,u\in S_b$) is minimized.
\end{problem}

Although \prob{MIN-UNCUT} is \NP{}-hard, it turns out that the graph $G_{ab}$ has a
special structure, which we call \df{almost bipartite}.

\begin{definition}
A graph $G=(V,E)$ is called \emph{almost bipartite} if it is a union
of a bipartite graph $H=(V_H,E_H)$ and two additional vertices $b, t$
whose edges may be incident to vertices of both partitions of $H$, that is,
$V=V_H\cup\{b\}\cup\{t\}$ and $E=E_H\cup E'$, where $E' \subseteq \{(b,
v) \mid v\in V\}\cup\{(t, v)\mid v\in V\}$.
\end{definition}

The bipartition is given by the fact that ``left'' (similarly, ``right'') terminals
of two lines can never be connected by an edge in $G_{ab}$. We show that \prob{MIN-UNCUT}
can be solved optimally for almost bipartite graphs.

\begin{figure}[t]
	\centering
    \includegraphics[width=\linewidth]{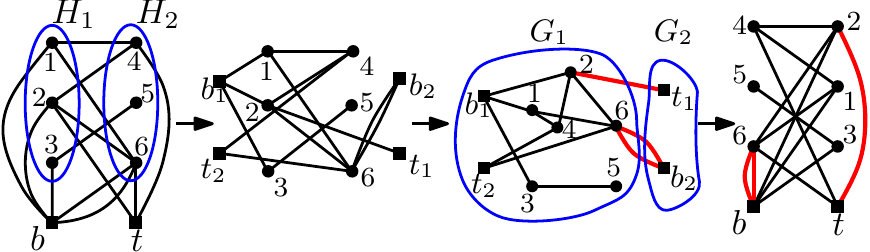}
  \caption{Solving \prob{MIN-UNCUT} on an almost bipartite graph.
    The maximum flow (minimum cut) with value $3$ results in vertex partitions $V^1_b=\{b_1, 4, 5, 6\}$,
  $V^1_t=\{t_2, 1, 2, 3\}$, $V^2_b=\{b_2\}$, and $V^2_t=\{t_1\}$. The optimal partition
  $S_b = \{b, 4, 5, 6\}, S_t = \{t, 1, 2, 3\}$ induces $3$ uncut edges $(b, 6), (b, 6), (t, 2)$.}
  \label{fig:ab}
\end{figure}

\begin{theorem}
\label{thm:uncut}
\prob{MIN-UNCUT} can be solved in polynomial time on almost bipartite graphs.
\end{theorem}
\begin{proof}
Almost bipartite graphs are a subclass of \emph{weakly bipartite graphs}~\cite{barahona83}.
Weakly bipartite graphs have no easy combinatorial characterization.
It is known that \prob{MAX-CUT} and \prob{MIN-UNCUT} can be solved in
polynomial time on weakly bipartite
graphs using the ellipsoid method~\cite{grotschel81}. However, the algorithm
might be not fast in practice. As mentioned in~\cite{grotschel81} ``it remains a challenging problem
to find a practically efficient method for the max-cut problem in weakly bipartite
graphs which is of a combinatorial nature and does not suffer from the drawbacks of
the ellipsoid method''. In the following we present such an algorithm
for almost bipartite graphs.

The special vertices $b$ and $t$ have to belong to different partitions of
$G_{ab}$.
We create a new graph $G'$ from $G_{ab}$. We split vertex $b$ into $b_1, b_2$ and
vertex $t$ into $t_1,t_2$ such that $b_1$ and $t_2$ are connected to the vertices of the first partition
$H_1$ of $H$, and $b_2$ and $t_1$ are connected to the second partition $H_2$. Formally,
for each edge $(b, v)\in E, v\in H_1$, we create an edge $(b_1, v)$; for each edge $(b, v)\in E, v\in H_2$,
we create an edge $(v, b_2)$. Similarly, edges $(v, t_1)$ are created for all $(t, v)\in E, v\in H_1$,
and edges $(t_2, v)$ are created for all $(t, v)\in E, v\in H_2$.
The construction is illustrated in Fig.~\ref{fig:ab}.

Now, for each edge $(u, v)$ of $G'$ we assign capacity $1$, and compute a maximum flow between the pair of
sources $b_1, t_2$ to the pair of sinks $b_2, t_1$. This can be done in
polynomial time with a maximum flow algorithm by introducing a
supersource (connected to $b_1$ and $t_2$) and a supersink
(connected to $b_2$ and $t_1$). Indeed, there is an integral maximum flow in $G'$.

A maximum flow corresponds to a maximum set of edge-disjoint paths starting at $b_1$ or $t_2$ and ending at $b_2$ or $t_1$.
Such a path corresponds to one of the following structures in
the original graph $G$: (i) an odd cycle containing vertex $b$ (a
cycle with an odd number of edges);
(ii) an odd cycle containing vertex $t$; (iii) an even path between
$b$ and $t$.

Note that if a graph has an odd cycle then at least one
of the edges of the cycle belongs to the same partition in any solution of \prob{MIN-UNCUT}.
The same holds for an even path connecting $b$ and $t$ in $G$ since $b$ and $t$ have to belong to different
partitions. Since the maximum flow corresponds to the edge-disjoint odd cycles and even paths in
$G$, the value of the flow is a lower bound for a solution of \prob{MIN-UNCUT}.

Let us prove that the value of the maximum flow in $G'$ is also an upper bound. By Menger's theorem,
the value of the maximum flow in $G'$ is the cardinality of the minimum edge cut separating sources and sinks.
Let $E^*$ be the minimum edge cut and let $G_1$ and $G_2$ be the correspondent disconnected subgraphs
of $G'$; see Fig.~\ref{fig:ab}. Notice that $G_1$ is a bipartite graph since $H\cap G_1$ is bipartite; vertex $b_1$ is only
connected to vertices of $H_1$ and vertex $t_2$ is only connected to vertices of $H_2$. Therefore, there is a
2-partition of vertices of $G_1$ such that $b_1$ and $t_2$ belong to different partitions; let us denote
the partitions $V^1_b$ and $V^1_t$. Similarly, there is a 2-partition of $G_2$ into $V^2_b$ and $V^2_t$
with $b_2 \in V^2_b$ and $t_1 \in V^2_t$. We combine these partitions so that
$S_b = \{b\} \cup \left(V^1_b \cup V^2_b\right) \setminus~\{b_1,
b_2\}$ and $S_t = \{t\} \cup \left(V^1_t \cup V^2_t\right) \setminus \{t_1, t_2\}$.
Note that $S_b$ and $S_t$ are the required partitions
of vertices of $G$ for \prob{MIN-UNCUT}, and the set of uncut edges is
$E^*$, which completes the proof of the theorem.
\end{proof}

\begin{wrapfigure}[10]{r}{.28\textwidth}
	\vspace{-5ex}
  \centering
    \includegraphics[]{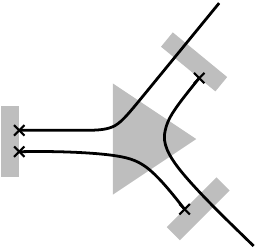}
  \caption{Example without consistent line directions.}
  \label{fig:claw}
\end{wrapfigure}
As a direct corollary, we get a polynomial-time algorithm for \PrMLCMP{} on caterpillars.
For which underlying networks can the algorithm be applied? Let $(G=(V, E), L)$ be an instance of
\PrMLCMP{}. We say that the lines $L$ have \df{consistent} directions on $G$ if the lines can
be directed so that for each edge $e \in E$ all lines $L_e$ have the same direction.
If the underlying graph is a path then we can consistently direct the lines from
left to right. Similarly, consistent line directions exist for ``left-to-right''~\cite{bekos08,argyriou09} and
``upward''~\cite{fink+pupyrev13} trees, that is, trees for which there is an embedding with all lines
being monotone in some direction. It is easy to test whether there are consistent
line directions by giving an arbitrary direction to some first line,
and then applying the same direction on all lines sharing edges with
the first line until all lines have directions or an inconsistency is found.
Hence, we get the following result.

\begin{theorem}
\label{thm:consist}
  \PrMLCMP{} can be solved in polynomial time for instances $(G, L)$ admitting
  consistent line directions.
\end{theorem}
\begin{proof}
Given consistent line directions, we assign top/bottom sides of each port as follows.
Consider a port of $u \in V$ corresponding to $(u, v) \in E$. Let $\pi_{uv} = (l_1 \dots l_p \dots l_q \dots l_{|L_{uv}|})$,
where $u$ is a terminal for the lines $l_1, \dots, l_p, l_q, \dots, l_{|L_{uv}|}$ and
$u$ is an intermediate station for the lines $l_{p+1}, \dots, l_{q-1}$.
We assume that the lines $l_1, \dots, l_p$ terminate at the top side of the port if the lines $L_{uv}$ are
directed from $u$ to $v$; otherwise, $l_1, \dots, l_p$ terminate at the bottom side of the port.

Let us consider a pair of lines $l, l'$ having a common subpath $P$ starting at $u$ and ending at $v$.
It is easy to see that for the terminal ports of $u$ and $v$ their top sides are located on the same side of $P$.
Hence, in our 2SAT model, we have only crossing formulas of type
$(l_u \vee l_v') \wedge (\neg l_u \vee \neg l_v')$ (apart from clauses
consisting of a single variable). Therefore, the graph $G_{ab}$ contains an edge $(l_u, l_v')$
for the pair of lines.

Let us show that $G_{ab}$ is almost bipartite. To this end, we prove that there is no odd cycle
containing only vertices $l_u$ for $u\in V, l \in L$. Suppose there is such a cycle $C$. Let $l^1_u, l^2_v$
be the first two lines in the cycle; a common subpath $P$ of $l^1_u$ and $l^2_v$ starts at $u$ and ends at $v$.
We may assume without loss of generality that the lines are directed from $u$ to $v$.
Consider the port at $u$ corresponding to the first edge $(u, u_1)$ of $P$. Observe that the direction of the lines
is ``aligned'' with the port; that is, the lines are directed from $u$ to $u_1$.
Now consider the port at $v$ corresponding to the last edge $(v, v_1)$ of $P$.
The direction of the lines is ``opposite'' to the port; that is, the lines directed from $v_1$ to $v$.
It is easy to see that for the next line $l^3_w$ in the cycle $C$ the direction of the lines is again ``aligned''
with the corresponding port of $w$. Moreover, for every line $l^{2k+1}$ the corresponding port is ``aligned''
with the direction of lines, and for every $l^{2k}$ the direction of lines is ``opposite'' to the port.
Hence, there cannot exist an odd cycle $C$.
\end{proof}

Note that there are examples of trees without consistent line directions; see
Fig.~\ref{fig:claw}.

\section{Practical Questions on MLCM}
\label{sec:practice}

Not every real-world transportation network meet the requirements implied by
our model. For example, a graph introduced in Section~\ref{sec:p-mlcm-p} is
not necessarily almost bipartite or some lines may be subpaths of another lines.
At the same time, lines are not necessarily simple paths as many metro maps have circular
or tree-like lines. Thus, the existing algorithms cannot be applied. We propose two
directions to for future work.

\begin{enumerate}
\item In many metro networks, there are just few
lines violating the required properties. We suggest to first create an instance with
the desired properties by deleting few (parts of) lines. Then, after applying our
algorithm, the deleted parts can be reinserted with as few crossings as possible.
As we show in Lemma~\ref{lm:insert}, a line can be inserted into an existing order
with the smallest number of introduced crossings. A number of possible questions/extensions
is possible in the direction. For example, how to find the ``best'' set of edges to remove?
An algorithm for insertion several lines optimally is also needed.

\item Although both \MLCM{} and \MLCM{} variants are \NP{}-hard, there is a hope
to construct fixed-parameter tractable algorithms (in addition to the one presented in Section~\ref{sec:mlcm-p-2sat}),
which are able to produce exact solutions for real-world instances. As pointed out in~\cite{nollenburgThesis},
the maximal number of parallel metro lines
per edge is reasonably small in practice. Hence, one could try to design a
fixed-parameter tractable algorithm for the metro-line crossing minimization
problem with respect to the parameter. We develop such an algorithm for the case
when an underlying network is a caterpillar, see Theorem~\ref{thm:fptk}. Designing an
algorithm for general graphs is an interesting open problem.
\end{enumerate}

We also observe that so far, the focus has been on the number of crossings and not on their
visualization, although two line orders with the same crossing number may look quite
differently~\cite{fink+pupyrev13}. Therefore, an important practical problem is the visual representation
of computed line crossings. In our opinion, crossings of lines should preferably be
close to the end of their common subpath as this makes it easier to recognize that the
lines do cross. It is not always possible to find an optimal solution in which
every pair of lines cross at the end of their common subpath, see~\cite{pupyrev11}.
Is there a compromised solution with a small number of crossings and reasonable
distribution of crossings?
For making a metro line easy to follow the important
criterion is the number of its bends. Hence, an interesting question is how to
sort metro lines using the minimum total number of bends.

\subsection{Fixed-Parameter Tractability of \MLCM{}/\MLCMP{} on Paths}
Let $k$ be the maximum number of lines per edge.

\begin{theorem}
\label{thm:fptk}
An optimal order of lines for \MLCM{}/\MLCMP{} can be computed in time $O((k!)^2|V|)$ if the
underlying network is a caterpillar.
\end{theorem}

\begin{proof}
Let $P=v_1v_2\dots v_n$ be a path with $n=|V|$
vertices. Denote $cr_{v_iv_{i+1}}(\pi)$ to be the minimum number of line crossings
on the subpath $v_1\dots v_i$ so that lines $L_{v_iv_{i+1}}$ form a permutation $\pi$
on the right port of vertex $v_i$; see Fig.~\ref{fig:fpt-path}. Symmetrically, denote
$cr_{v_{i+1}v_i}(\pi)$ to be the minimum number of line crossings on the subpath
$v_1\dots v_{i+1}$ so that lines $L_{v_iv_{i+1}}$ form a permutation $\pi$
on the left port of vertex $v_{i+1}$. Also denote $cr_{v_i}(\pi)$ to be the minimum
number of line crossings on the subpath $v_1\dots v_i$ so that lines $L_{v_{i-1}v_i} \cap L_{v_iv_{i+1}}$
form a permutation $\pi$ on the left port of $v_i$. Since vertex crossings
are not allowed in our model, $cr_{v_i}(\pi)$ also corresponds to the optimal order on the right port of $v_i$.
If the problem being considered is \MLCM{} then we call an order of lines $\pi$ is \df{valid} on the left (right) port
of vertex $v_i$ if (i) the order does not induce vertex crossings at $v_i$. If we consider \MLCMP{} then
the order $\pi$ is \df{valid} if (i) holds and (ii) the line start/end at the bottommost or topmost position
of their terminals; see Fig.~\ref{fig:fpt-path2}. Notice that the algorithms for \MLCM{} and \MLCMP{} differ only in the
definition of a valid permutation.

\begin{figure}[t]
    \centering
    \subfigure[]{
    \includegraphics[height=3cm]{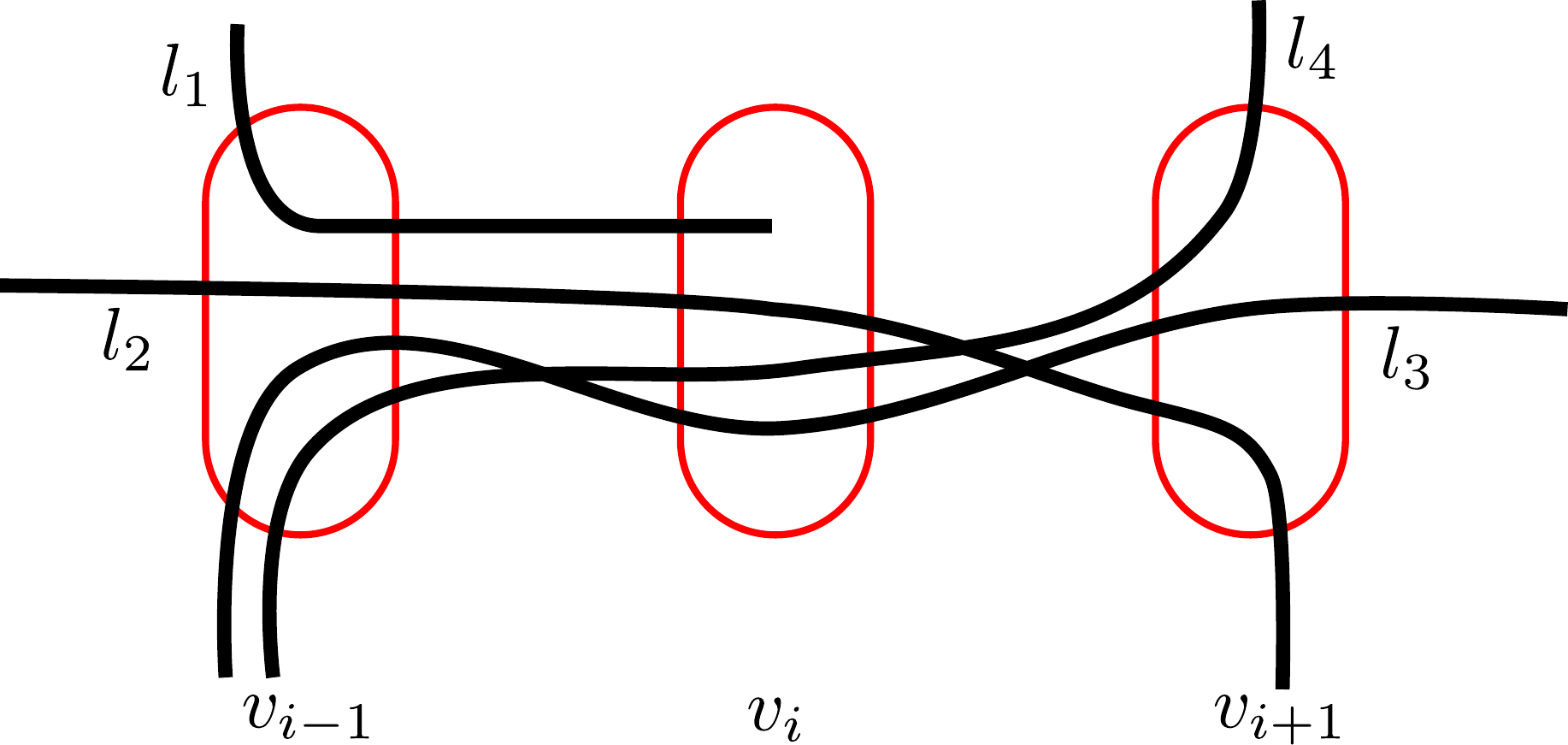}
    \label{fig:fpt-path}}
~~~~~~~~
    \subfigure[]{
    \includegraphics[height=3cm]{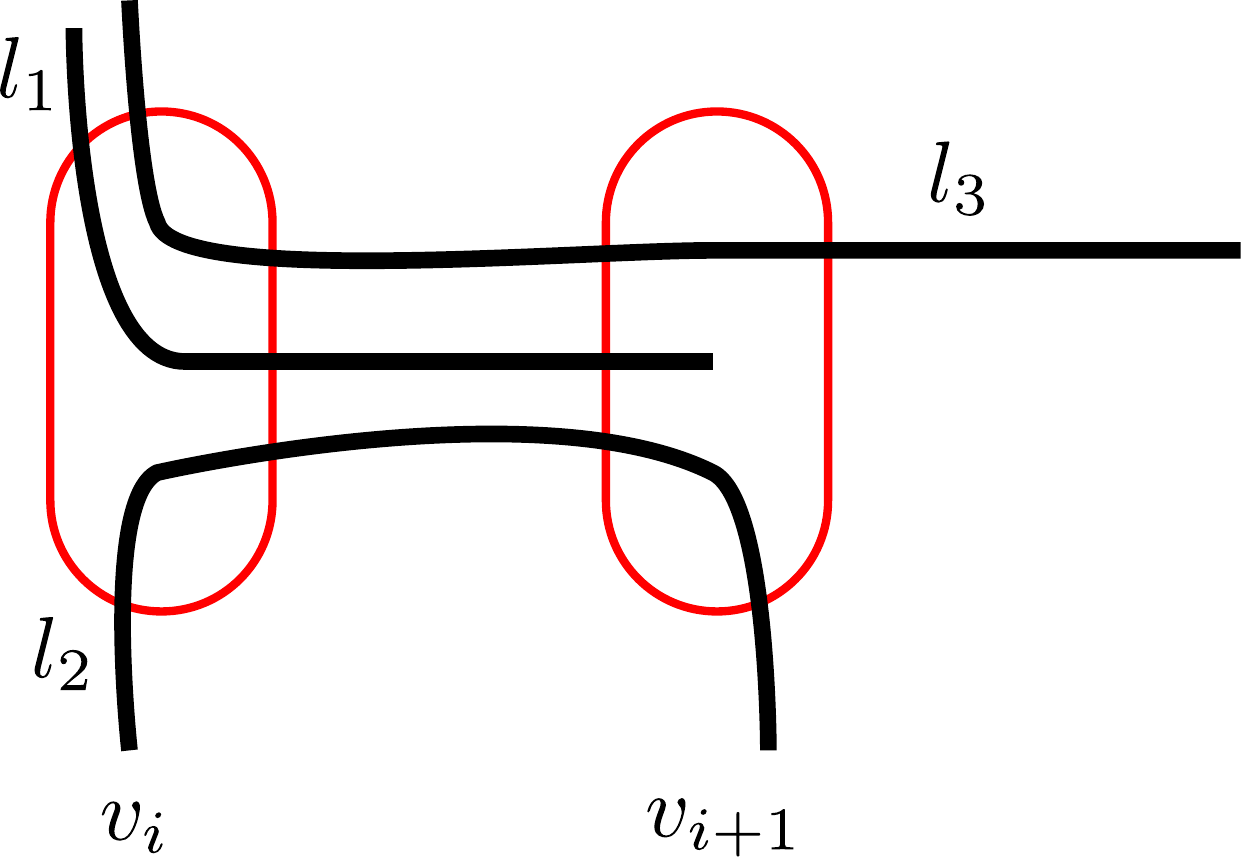}
    \label{fig:fpt-path2}}
    \caption{(a)~Solving \MLCM{}/\MLCMP{} on a path. The order of paths $L_{v_{i-1}v_i}$ on the
    right port of $v_{i-1}$ is $(l_1, l_2, l_3, l_4)$, and on the left port of $v_i$~-- $(l_1, l_2, l_4, l_3)$.
    (b)~The permutation $(l_3, l_1, l_2)$ is valid on the right port of $v_i$ for both \MLCM{} and \MLCMP{}.
    The permutation $(l_3, l_1, l_2)$ on the left port of $v_{i+1}$ is valid for \MLCM{} but not for \MLCMP{} since
    line $l_1$ has to terminate topmost or bottommost at $v_{i+1}$.}
\end{figure}

We compute values $cr(\pi)$ using dynamic programming by iterating over $P$ from $v_1$ to $v_n$.
Clearly, $cr_{v_1}(\pi) = 0$ for the only ``empty'' permutation $\pi=()$ since no lines
passes through vertex $v_1$. Let us describe how to compute intermediate values of $cr$.
\begin{itemize}
\item Let $(v_i, v_{i+1})$ be the current edge. We show how to compute $cr_{v_iv_{i+1}}$ from $cr_{v_i}$.

    For a valid permutation $\pi$ on the right port of $v_i$, we set
    $cr_{v_iv_{i+1}}(\pi) = cr_{v_i}(\sigma)$, where $\sigma$ is a subpermutation of $\pi$ comprised of the lines
    $L_{v_{i-1}v_i} \cap L_{v_iv_{i+1}}$. The step requires $O(k!)$ time.

\item Let $(v_i, v_{i+1})$ be the current edge. We show how to compute $cr_{v_{i+1}v_i}$ from $cr_{v_iv_{i+1}}$.

    It is easy to see that if lines form a permutation $\pi$ on the right port of $v_i$ and a permutation $\pi'$
    on the left port of $v_{i+1}$ then the number of crossings they make is exactly the number of
    inversions between $\pi$ and $\pi'$, that is, the number of pair $a, b$ with $\pi = (\dots a \dots b \dots)$ and
    $\pi' = (\dots b \dots a \dots)$. Let $inv(\pi, \pi')$ be a number of inversions between $\pi$ and $\pi'$. Then
    $cr_{v_{i+1}v_i}(\pi) = \min_{\pi'} (cr_{v_iv_{i+1}}(\pi') + inv(\pi, \pi'))$, where $\pi$ is a valid permutation
    of lines $L_{v_iv_{i+1}}$ on the left port of $v_{i+1}$; the minimum is taken over all valid permutations $\pi'$ of $L_{v_iv_{i+1}}$ on the right port of $v_{i}$. The step requires $O((k!)^2)$ time.

\item Let $(v_i, v_{i+1})$ be the current edge. We show how to compute $cr_{v_{i+1}}$ from $cr_{v_{i+1}v_i}$.

    For all permutations $\pi$ of lines $L_{v_{i-1}v_i} \cap L_{v_iv_{i+1}}$, let
    $cr_{v_{i+1}}(\pi) = \min_{\sigma} cr_{v_{i+1}v_i}(\sigma)$, where the minimum is taken over all valid permutations $\sigma$ on the left port of $v_{i+1}$ so that $\pi$ is a subpermutation of $\sigma$. The step can be done in
    $O(k!)$ time since the number of different permutations $\sigma$ is at most $k!$

\end{itemize}

It is easy to see that the minimum number of crossings for the \MLCM{} and \MLCMP{} problems is $cr_{v_n}(\pi)$
for the ``empty'' permutation $\pi=()$.
\end{proof}

\subsection{Optimal Insertion of a Line into an Existing Solution}
In this section, we explore a simple heuristic for computing line orders. The heuristic
works iteratively by inserting lines into an existing order. Let $l_1, \dots, l_{|L|}$ be
the input lines. The heuristic makes $|L|$ iteration, and on the iteration $i$ line $l_i$ is inserted
into the corresponding line orders. It turns out that every line can be inserted optimally;
that is, we can minimize $cr_i - cr_{i-1}$ for $1 < i \le |L|$, where $cr_i$ is the number of
crossings in a solution with lines $l_1, \dots, l_i$. Notice that the insertion algorithm works
for both \MLCM{} and \MLCMP{} models.

\begin{lemma}
\label{lm:insert}
Let $G = (V, E)$ be an embedded graph, $L$ be a set of lines on $G$, and
let $\pi_{uv}$ be a fixed order of the lines for all $u\in V$ and $(u, v) \in E$.
There is a polynomial-time algorithm for insertion a line $l$ into the existing
order so that the number of newly introduced crossings is minimized.
\end{lemma}

\begin{proof}
Let $l = v_1v_2\dots v_k$ with $v_i \in V$. Create a graph $H=(U, W)$ in which a vertex is
a ``gap'' between lines in every port traversed by $l$, and an edge is a ``valid'' route for $l$.
Formally, for every $v_i, 1\le i < k$ create $|\pi_{v_iv_{i+1}}|+1$ vertices referred to as $V_i^r \subset U$;
similarly, for every $v_i, 1 < i \le k$ create $|\pi_{v_iv_{i-1}}|+1$ vertices referred to as $V_i^l \subset U$
(here, $|\pi_{v_iv_{i+1}}|$ is the number of lines in the sequence $\pi_{v_iv_{i+1}}$).
An edge $(u_1, u_2)$ is added to $W$ if (i) $u_1 \in V_i^r, u_2 \in V_{i+1}^l$ or $u_1 \in V_i^l, u_2 \in V_i^r$,
and (ii) the line $l$ passing through $u_1$ and $u_2$ does not violate requirements of the model.
For the condition (ii) we check whether $l$ introduces vertex crossings (for both variants \MLCM{} and \MLCMP{})
and the periphery condition is satisfied (only for \MLCMP{}).
Then assign weights for the edges of $H$ as the number of newly added crossings with the line $l$; see
Fig.~\ref{fig:insertion}.

\begin{figure}[ht]
  \begin{center}
    \includegraphics[height=3cm]{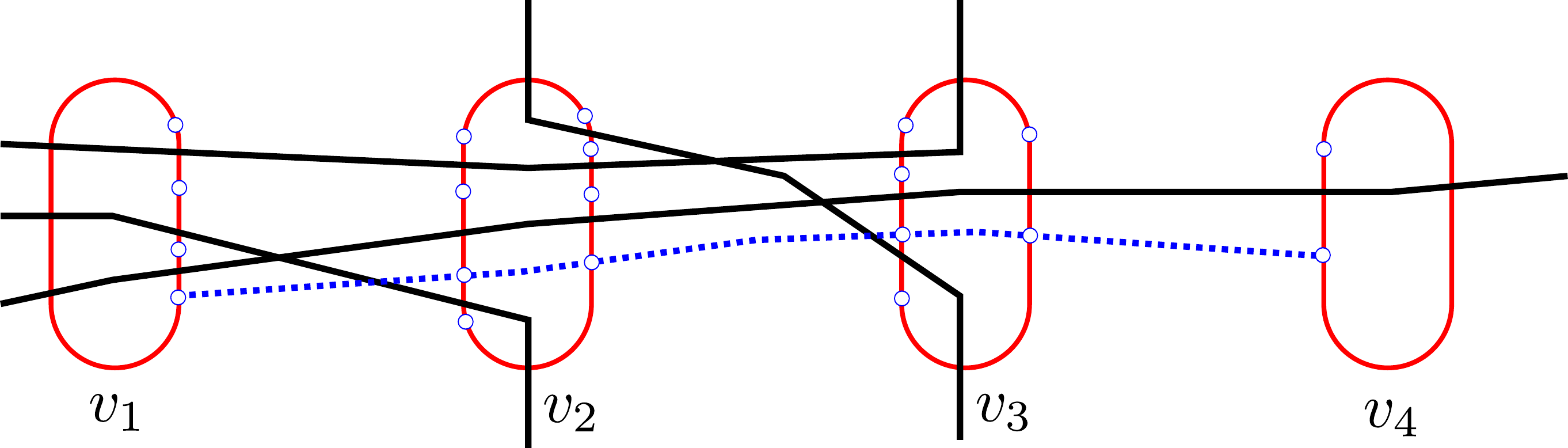}
  \end{center}
  \caption{Insertion of a line $l=v_1v_2v_3v_4$ (blue) into an existing order in the \MLCMP{} model.
  Vertices of the graph $H$ are shown as blue circles.}
  \label{fig:insertion}
\end{figure}

It is easy to see that an insertion of the line $l$ into the existing line orders corresponds to a path on the graph $H$.
Hence, in order to optimally insert $l$ we find a shortest path on $H$ connecting a vertex from $V_1^r$
(terminal port of $l$ at $v_1$) to a vertex from $V_k^l$ (terminal port of $l$ at $v_k$). Again, a source and
a destination for the shortest path should be chosen so that the condition (ii) is satisfied.
The vertices of the shortest path correspond to the desired positions of $l$ in the existing line orders.
\end{proof}

\section{Conclusion and Open Problems}
We proved that \MLCM{} is \NP{}-hard and presented an $O(\sqrt{\log
|L|})$-approximation algorithm for \MLCMP{}, as well as an exact
$O(|L|^3)$-time algorithm for \PrMLCMP{} on instances with consistent
line directions. We also suggested polynomial-time
algorithms for crossing-free solutions for \MLCM{} and \MLCMP{}.
From a theoretical point of view, there are still many interesting
open problems:
\begin{compactenum}
\item Can we derive
an approximation algorithm for \MLCM{}?

\item Is there a constant-factor approximation algorithm for \MLCMP{}?

\item What is the complexity status of \PrMLCM{}/\PrMLCMP{} in general?
\end{compactenum}

On the practical side, the visualization of the computed line
crossings is a possible future
direction. So far, the focus has been on the number of
crossings, although two line orders with the same crossing number may look quite
differently \cite{fink+pupyrev13}.
The question on how to visualize the crossings is especially important for curvy metro maps~\cite{fhnrsw-dmmbc-GD12}.
For example, a metro line is easy to follow if it has few bends. Hence, an
interesting question is how to visualize metro lines using the minimum total
number of bends.

\paragraph{Acknowledgments.}
We thank Martin N{\"o}llenburg, Jan-Henrik Haunert, Joachim Spoer\-hase, Lukas Barth, Stephen Kobourov,
and Sankar Veeramoni for discussions about variants of the metro-line crossing minimization problem.
We are especially grateful to Alexander Wolff for help with the paper.

\bibliographystyle{abbrv}
\bibliography{mlcm}

\end{document}